\documentclass[sigplan, 10pt]{acmart}

\setcopyright{none}
\renewcommand\footnotetextcopyrightpermission[1]{}
\acmConference[unpublished]{manuscript under review}{Feb. 2022}{somewhere}
\makeatletter
\renewcommand\@formatdoi[1]{\ignorespaces}
\makeatother

\usepackage[utf8]{inputenc}
\usepackage{graphicx}
\usepackage{tabularx}
\usepackage{amsmath}
\usepackage{algorithm}
\usepackage{algpseudocode}
\usepackage{multirow}

\def \rd {\textit{rd}}

\newtheorem{theorem}{Theorem}

\newtheorem{lemma}{Lemma}

\newtheorem{definition}{Definition}

\AtBeginDocument{%
  \providecommand\BibTeX{{%
    \normalfont B\kern-0.5em{\scshape i\kern-0.25em b}\kern-0.8em\TeX}}}

\setcopyright{acmcopyright}
\copyrightyear{2018}
\acmYear{2018}
\acmDOI{10.1145/1122445.1122456}

\acmConference[Woodstock '18]{Woodstock '18: ACM Symposium on Neural
  Gaze Detection}{June 03--05, 2018}{Woodstock, NY}
\acmBooktitle{Woodstock '18: ACM Symposium on Neural Gaze Detection,
  June 03--05, 2018, Woodstock, NY}
\acmPrice{15.00}
\acmISBN{978-1-4503-XXXX-X/18/06}

\begin{document}

\title{Beyond Time Complexity: Data Movement Complexity Analysis for Matrix Multiplication}


\author{Wesley Smith}
\email{wsmith6@cs.rochester.edu}
\orcid{1234-5678-9012}
\affiliation{
  \institution{University of Rochester}
  \city{Rochester}
  \state{NY}
  \country{USA}
}

\author{Aidan Goldfarb}
\email{agoldfa7@u.rochester.edu}
\affiliation{
  \institution{University of Rochester}
  \city{Rochester}
  \state{NY}
  \country{USA}
}

\author{Chen Ding}
\email{cding@cs.rochester.edu}
\affiliation{
  \institution{University of Rochester}
  \city{Rochester}
  \state{NY}
  \country{USA}
}

\renewcommand{\shortauthors}{Trovato and Tobin, et al.}

\begin{abstract}

Data movement is becoming the dominant contributor to the time and energy costs of computation across a wide range of application domains.  However, time complexity is inadequate to analyze data movement. This work expands upon Data Movement Distance, a recently proposed framework for memory-aware algorithm analysis, by 1) demonstrating that its assumptions conform with microarchitectural trends, 2) applying it to six variants of matrix multiplication, and 3) showing it to be capable of asymptotically differentiating algorithms with the same time complexity but different memory behavior, as well as locality optimized vs. non-optimized versions of the same algorithm. In doing so, we attempt to bridge theory and practice by combining the operation count analysis used by asymptotic time complexity with per-operation data movement cost resulting from hierarchical memory structure. Additionally, this paper derives the first fully precise, fully analytical form of recursive matrix multiplication's miss ratio curve on LRU caching systems. Our results indicate that the Data Movement Distance framework is a powerful tool going forward for engineers and algorithm designers to understand the algorithmic implications of hierarchical memory.
\end{abstract}




\maketitle
\pagestyle{plain}

\section{Introduction}
In exascale computing, the cost of data movement exceeds that of computation \cite{dallyExascale}: as such, data movement is a key factor in not only system performance but also in the computer science community's growing responsibility to address computing's role in the climate crisis.  Optimizing a program or a system for locality is difficult, as modern memory systems are large and complex.  For portability, we should not program data movement directly, but we should be aware of its cost.  However, there does not yet exist a single quantity that can characterize the effect of locality optimization at the program or algorithm level. Standard techniques for understanding algorithm-hierarchy interactions, like miss ratio analysis, yield insight for algorithm designers with a target set of machine parameters in mind, but do not allow for general understanding of an algorithm's intrinsic data movement for an arbitrary target machine.

% CD： Here is a different way to say what's said above:
% Locality in data access is of paramount importance, but as programmers we have no software measure to quantify data movement.  We have to run a program on a machine to measure.  Yet after measuring, we are often left with a set of numbers with unclear relationship between them. 

The actual effect of cache usage depends on all components of a program and also its running environment.  Assumptions about a machine may be wrong, imprecise, or soon obsolete.  A program may run on a remote computer in a public or commercial computing center with limited information about its memory system.  Auto-tuning can select the best parameters for a given system, but it is difficult to tune if a system is shared. 

In a recent position paper, \citet{DingS:MEMSYS21} defined an abstract measure of memory cost called Data Movement Distance (DMD).  Memory complexity is measured by DMD in the same way time complexity is by operation count.  They showed results for two types of data traversals with only constant factor differences in DMD.

This paper presents DMD analysis for different approaches to matrix multiplication.  It differs from past work in several ways.  First, unlike practical analysis, i.e. those based on miss ratios, DMD analysis is asymptotic and machine agnostic.  Second, unlike I/O complexity, DMD analysis includes the effect of a cache hierarchy. In addition, the derivation is radically different from past solutions. For example, efficient algorithms often make use of temporaries that are dynamically allocated.  They share cache, but the cache sharing is not analyzed by past complexity analysis.  It is measured by practical analysis in concrete terms, i.e. cache misses, not asymptotic terms.

\section{Main Contributions}
The focus of this work is exploring the algorithmic implications of hierarchical memory systems by fleshing out the memory-aware algorithm analysis framework Data Movement Distance (DMD) introduced in \cite{DingS:MEMSYS21}. The main contributions are as follows:
\begin{itemize}
    \item Derivation and empirical validation of the first fully precise analytical form of recursive matrix multiplication's miss ratio on LRU caching systems,
    \item Application of the DMD framework for memory-aware algorithm analysis to four variants of matrix multiplication,
    \item Exploration of the effects of locality optimizations on the previous results,
\end{itemize}
In addition, we expand the motivation for and justification of the DMD framework by demonstrating the following:
\begin{itemize}
    \item microarchitectural trends in cache memory conform with the framework's assumptions,
    \item DMD is capable of asymptotically differentiating algorithms with the same time complexity as a result of their memory behavior, as well as asymptotically differentiating between locality optimizations
\end{itemize}
Taken together, we believe that the results of our analyses indicate that the DMD framework is a powerful tool going forward for engineers and algorithm designers to understand the algorithmic implications of hierarchical memory.
\section{Background and Motivation}
\subsection{Locality Concepts}
The locality concept most central to this work is \textit{reuse distance} (RD) \cite{Yuan+:TACO19}. Reuse distance, or LRU stack distance~\citep{Mattson+:IBM70}, characterizes an individual memory access by counting the number of distinct memory locations accessed by the program between the most recent previous use of that memory location and the current use. 

For example, let letters denote distinct memory locations in the following access trace:
\begin{align*}
    abbca
\end{align*}
In this example, the reuse distance of the second access to $b$ is 1, as only $b$ occurs in the window from position 2 to position 3, while the reuse distance of the second access to $a$ is 3, as $a,b,c$ all occur in the window from position 1 to position 5.

Reuse distance and miss ratio for fully-associative LRU cache are interconvertible, with their relation as follows:
\begin{align*}
    MR(c) = P(rd > c)
\end{align*}
Accesses with reuse distance greater than $c$ are misses in LRU caches of size $c$ or less. So, reuse distance distributions and miss ratio curves are the same information.
\subsection{Data Movement Distance}
The most ubiquitous technique for algorithm cost analysis is asymptotic time complexity, which measures operation count as a function of input size. However, on machines with hierarchical memory, execution cost will scale with input size \textit{faster} than time complexity would indicate, because as data size increases, more program data must be stored in large, slow hierarchy components: the cost of data movement scales with input size as well. Data Movement Distance (DMD) is a novel framework for memory-aware algorithm analysis proposed by Ding and Smith \cite{DingS:MEMSYS21} in which operation count is combined with per-access data movement cost by considering the algorithm's reuse distance distribution. 

Because data movement cost varies across machines, the DMD framework includes an abstracted version of a memory hierarchy, termed the \textit{geometric stack}, on which the behavior of algorithms is considered. The geometric stack can be understood to be an infinite-level memory hierarchy in which each level stores a single datum. The cost of accessing the datum at level $n$ is $\sqrt{n}$. DMD for a program is then defined as follows:
\begin{definition}[Data Movement Distance]
\label{def:dmd}
For a program $p$ with data accesses $a_i$, let the reuse distance of $a_i$ be $d_i$.  The DMD for $p$ under caching algorithm $A$ is
\begin{align*}
    DMD(p) = \sum_i \sqrt{d_i}
\end{align*}
\end{definition}
In words, a program's DMD is the sum of the square roots of its memory accesses' reuse distances. As such, we arrive at our first theorem:
\begin{theorem}[DMD bounds]
Let the time complexity of program $p$ be $O(f(n))$ and let its space complexity be $O(g(n))$. Then
\begin{align*}
    f(n) \leq O(DMD(p)) \leq f(n) \cdot \sqrt{g(n)}
\end{align*}
\end{theorem}
\begin{proof} It suffices to note that this program will have $f(n)$ memory accesses, the minimum value a reuse distance can take is 1, and the maximum value a reuse distance can take is $g(n)$ (data size). Then, Definition \ref{def:dmd} yields the above.
\end{proof}
An intuitive explanation for the $\sqrt{n}$ cost function is that it is the distance the data must travel if we represent the infinite-level hierarchy as a series of concentric 2-D shapes (such as circles) and let area represent capacity. In the following section, we will discuss the microarchitectural trend that this cost function reflects. In Ding and Smith's formulation for DMD, cache replacement policy is a parameter, however in the derivations in this paper we will use LRU replacement.

Throughout this paper we will use the asmyptotic equivalence notation $\sim()$, which is identical to big-O notation except it retains primary factor coefficients.
\subsection{Why $\sqrt{n}$?}
At the core of the DMD framework is the notion that we need a relationship between stack position, or reuse distance, and cost. Ding and Smith \cite{DingS:MEMSYS21} use $\sqrt{n}$ with the argument that it reflects physical memory layout. We expand on that here by demonstrating that $\sqrt{n}$ also reflects the cost of memory access on modern architectures.
\begin{figure}
    \centering
    \includegraphics[width=\columnwidth]{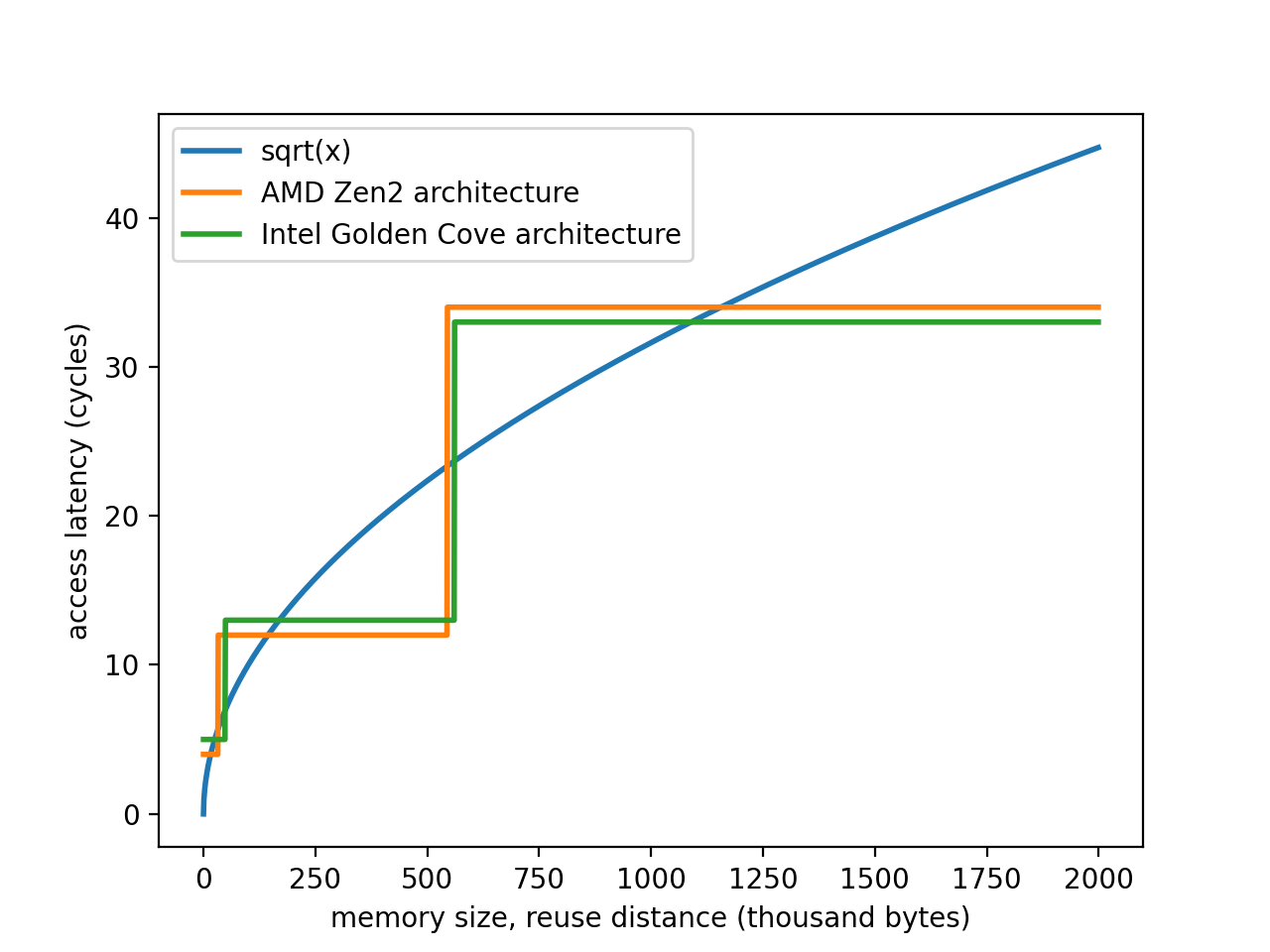}
    \caption{Access latencies of modern microarchitectures as a function of stack position}
    \label{fig:latencies}
\end{figure}
Figure \ref{fig:latencies} demonstrates access latency as a function of stack position for cache sizes and latencies corresponding to the AMD Zen2 and Intel Golden Cove architectures (numbers from \cite{cutress_2019}) up to $\approx2$MB data size. As stack position increases, data is forced into slower caches and its access latency increases.  While Fig. \ref{fig:latencies} assumes LRU caching and thus does not capture hardware optimizations, it is clear that the general trend in these architectures' latencies is well captured by the $\sqrt{x}$ family of functions. 

This relationship has been observed in other contexts as well: Yavitz et al. \cite{Yavitz+14} show that access latency scales with the square root of cache size, and Cassidy and Andreou \cite{Cassidy+TC11} demonstrate that energy cost behaves similarly. These results, and thus the square root concept, have been used in the design of cutting edge memory systems and techniques, such as Tsai et al.'s Jenga \cite{Tsai+ISCA17}.

\subsection{Applications}

DMD is to measure hierarchical locality.  
A program has \emph{hierarchical locality} if it makes use of a \emph{cache hierarchy}, where there is has more than a single level of cache, and the cache size and organization may vary from machine to machine.  Such cache hierarchies are the norm on today's machines.    

As a contrast, consider single-cache locality which means programming to utilize a cache of a specific size.  Single-cache locality is unreliable for two reasons.  The first is portability.  The actual cache usage depends on the choice of programming languages, compilers, and target machines.  The second problem is environmental, e.g. interference from run-time systems and from peer programs that share the same cache.  Single-cache locality is not a robust program property because of these two sources of uncertainty. 

Hierarchical locality is portable and elastic.  
It is independent of implementation, and it runs well in a shared environment.  Cache oblivious algorithms, e.g. recursive matrix multiplication, were developed for  hierarchical locality.  However, the hierarchical effect was not formalized or quantified in prior work.  Next, we use DMD to analyze this effect.  

% A better algorithm would improve locality regardless of which programming language, compiler, or target machine is used in implementation.  
% The second property has no commonly accepted name.  We call it \emph{elasticity}, which means that an algorithm or a program is capable of running well in a shared environment.  Since the available resource may change dynamically, being elastic means that its memory use can be ``stretched'' or ``compressed'', hence the notion of elasticity.  

\section{Recursive Matrix Multiplication Miss Ratio Analysis}
In this section we derive what is to our knowledge the first fully precise analytical form of recursive matrix multiplication's (RMM) miss ratio curve on LRU memory. Previously, RMM's asymptotic cache behavior has been derived \cite{Prokop1999}, but we derive the precise, numeric miss ratio for all cache sizes and matrix sizes and validate our model's accuracy against instrumented executions.

We will derive miss ratio by deriving the distribution of reuse distances incurred by RMM. Reuse distance and miss ratio have been shown to be interconvertible \cite{Yuan+:TACO19}, meaning they are the same information.

At a high level, our approach will be as follows:
\begin{enumerate}
    \item decompose the recursive algorithm into its canonical tree structure,
    \item split memory accesses into accesses to temporaries and accesses to input matrices,
    \item derive symbolic representations of each type of memory access pattern at each level of the tree,
    \item iterate over the entire tree to create the full reuse distribution.
\end{enumerate}
The full derivation is too lengthy to include in this paper; the interested reader may find a more detailed version at \textbf{github.com/anon}. We will present the lemmas and theorems that constitute the majority of the contribution of this section as well as the skeleton of the full derivation, but will elide in-depth proofs of many of the derived relationships. We demonstrate correctness by implementing our model and showing its functional equivalence to an instrumented version of RMM.
\begin{verbatim}
Function rmm(A,B):
    n = A.rows
    let C be a new nxn matrix
    if n == 1:
        C11 = A11 * B11
    else:
        C11 = rmm(A11, B11) + rmm(A12, B21)
        C12 = rmm(A11, B12) + rmm(A12, B22)
        C21 = rmm(A21, B11) + rmm(A22, B12)
        C22 = rmm(A21, B12) + rmm(A22, B22)
    return C
\end{verbatim}
 The above contains the pseudocode for standard recursive matrix multiplication, and is the form for RMM that we will analyze. Note that this is RMM at its most basic: we don't consider optimizations such as using a base case larger than $1\times1$ or temporary reclamation and reuse. We will consider such optimizations when analyzing the algorithm's data movement distance, but for our miss ratio/reuse distance derivation we analyze RMM in its simplest form. Here, each call to multiply $N \times N$ matrices decomposes into eight recursive calls, each multiplying $\frac{N}{2} \times \frac{N}{2}$ matrices. After each pair of recursive calls, there is an addition step.
\begin{figure}[h!]
    \centering
    \includegraphics[width=\columnwidth]{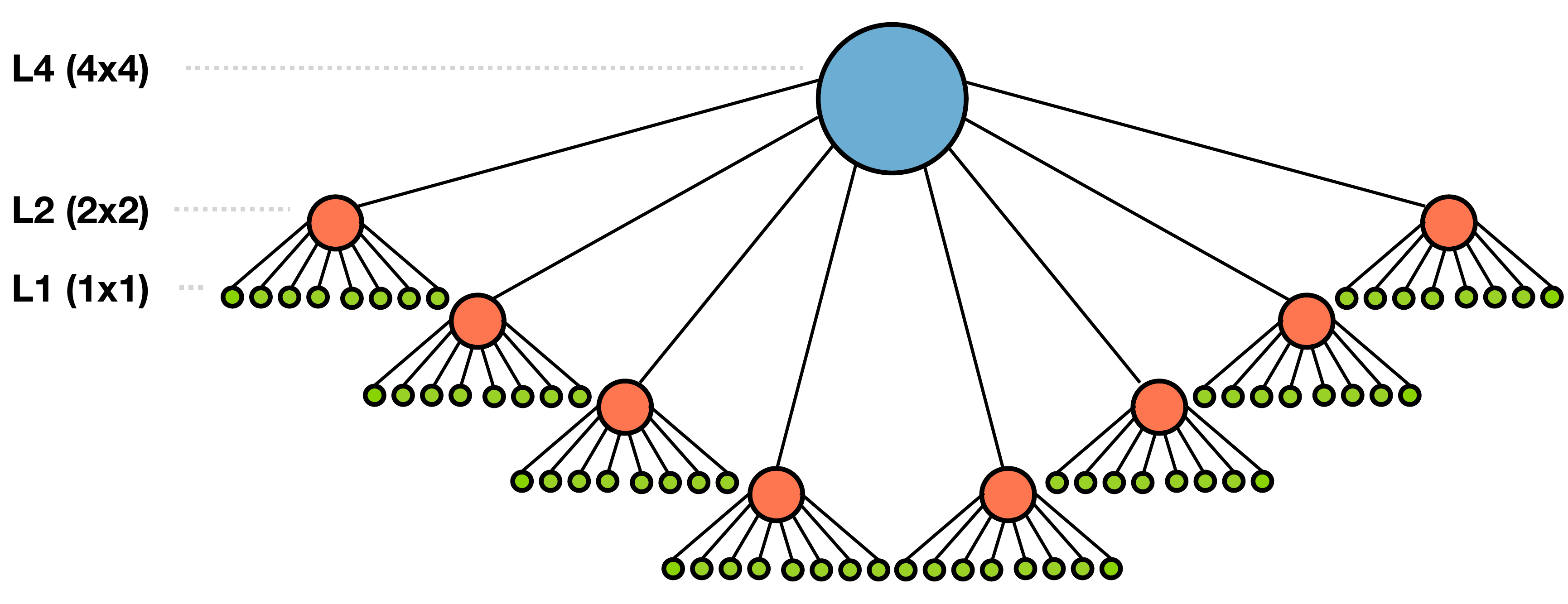}
    \caption{The tree structure of recursive matrix multiplication}
    \label{fig:tree1}
\end{figure}

Figure \ref{fig:tree1} contains the decomposition of a 4x4 multiplication. The blue node, marked L4, represents a 4x4 multiplication, the orange nodes, marked L2, represent 2x2 multiplications, and the green nodes, marked L1, represent 1x1 multiplications. Execution is a depth-first, left-to-right traversal of this tree. In the rest of this derivation, we will use the notation $LN$ to refer to a node in the tree that is multiplying $N \times N$ matrices.
\subsection{Temporaries}
We first analyze the behavior of RMM's temporary usage:
\begin{definition}[Temporary Count]
\label{def:tempcount}
Let $T_N$ represent the number of temporaries introduced in an $LN$ call to recursive matrix multiplication, i.e. multiplying NxN matrices. Then 
\begin{align*}
    T_N = \sum_{i=0}^{log_2(N)}8^{log_2(N)-i}\cdot (2^i)^2 = N^2\cdot(2N-1).
\end{align*}
\end{definition}
Having a uniform branching factor makes deriving the node count at each level trivial:
\begin{definition}
\label{def:nodecount}
In an $NxN$ recursive matrix multiplication, the number of $LX$ nodes is as follows:
\begin{align*}
    \# LX = 8^{log_2(N) - log_2(X)} = \frac{N^3}{X^3}
\end{align*}
\end{definition}
The next step is to exploit the tree's symmetry to understand how much repetition there is in the reuse distances of temporaries:
\begin{lemma}[Temporary Symmetry]
\label{lemma:symmetry}
Let $F_T(i,j,N,a)$ denote the reuse distance of the $(i,j)$-th element of the $a$-th temporary matrix introduced at $LN$. Then
\begin{align*}
    F_T(i,j,N,a) = F_T(i,j,N,(a \hspace{0.1cm} \% \hspace{0.1cm} 2))
\end{align*}
\end{lemma}
Lemma \ref{lemma:symmetry} has a straightforward high-level interpretation: at $LN$, there are only at most two $N\times N$ matrices' worth of temporaries with different reuse distances. These correspond to the first and second elements in one of the additions in the RMM pseudocode at the beginning of Section 4, with each addition group introducing temporaries with identical behavior.

We will denote these two temporary matrices $DT_{1,N}$ and $DT_{2,N}$. We elide their derivations, but the approach involved a mix of purely analytical analysis of RMM's tree structure and pattern extraction from empirical results from an instrumented version of RMM. Their forms are as follows, where ellipses indicate the same value in all columns or values decreasing by 1 per column:
\begin{lemma}[First Temporary Matrix]
\label{lemma:temp1}
Let $DT_{1,N}$ denote the matrix of reuse distances of temporaries introduced at level N in the first node of an addition group. Then
\footnotesize
\begin{align*}
    &\hspace{1.0cm}DT_{1,N} = 
    \\&\begin{bmatrix}
    d_1 & ...& d_1 - (\frac{n}{2} - 1)& d_2 - (\frac{n}{2})^2 + \frac{n}{2}  & ... & d_2 - (\frac{n}{2})^2 + 1\\
    d_1& ...& d_1 - (\frac{n}{2} - 1)& d_2 - (\frac{n}{2})^2 + n & ... & d_2 - (\frac{n}{2})^2 + \frac{n}{2} + 1\\
    ...&...&...&...&...&...\\
    d_1& ...& d_1 - (\frac{n}{2} - 1)& d_2 & ... & d_2 - \frac{n}{2} - 1\\
    d_1 - \phi(n)& ...&{\scriptscriptstyle d_1 - \phi(n) - (\frac{n}{2} - 1)}&{\scriptscriptstyle d_2 - \delta(n) - (\frac{n}{2})^2 + \frac{n}{2}  }& ... & {\scriptscriptstyle d_2 - \delta(n) - (\frac{n}{2}^2) + 1}\\
    d_1 - \phi(n)& ...&{\scriptscriptstyle d_1 - \phi(n) - (\frac{n}{2} - 1)}&{\scriptscriptstyle d_2 - \delta(n) - (\frac{n}{2})^2 + n }& ... & {\scriptscriptstyle d_2 - \delta(n) - (\frac{n}{2})^2 + \frac{n}{2} + 1}\\
    ...&...&...&...&...&...\\
    d_1 - \phi(n)& ...&{\scriptscriptstyle d_1 - \phi(n) - (\frac{n}{2} - 1)}& d_2 - \delta(n) & ... & d_2 - \delta(n) - \frac{n}{2} - 1\\
    \end{bmatrix}
\end{align*}
\normalsize
where 
\begin{align*}
    \scriptstyle \delta(N) &=\scriptstyle N^3\\
    \scriptstyle \phi(N) &=\scriptstyle N^3 - \frac{N^2}{2}\\
    \scriptstyle d_1 = \lfloor2D_N - (2\cdot&\scriptstyle(T_{\frac{N}{2}} - 2((\frac{N}{2})^2 - 1))\rfloor\\
    \scriptstyle d_2 = \lfloor2D_N - (4T_{\frac{N}{2}} + 2(\frac{N}{2})^2 &\scriptstyle- 2(\frac{N}{2} - 1) - (2(\frac{N}{2})^2 - \frac{N}{2}))\rfloor
\end{align*}
\end{lemma}
\begin{lemma}[Second Temporary Matrix]
\label{lemma:temp2}
Let $DT_{2,N}$ denote the matrix containing the reuse distances of the temporaries introduced at level N in the second node of an addition group. Then
\footnotesize
\medskip \begin{align*}
     &\hspace{1.0cm}DT_{2,N} =
    \\&\begin{bmatrix}
    d_3 &  ...& d_3 & d_4  & ... & d_4 \\
    d_3 + n &  ...& d_3 + n & d_4 + \frac{3N}{2} & ... & d_4 + \frac{3N}{2} \\
    ...&...&...&...&...&...\\
    d_3 + n(\frac{n}{2} - 1) & ...& d_3 + n(\frac{n}{2} - 1) & d_4 + (\frac{n}{2} - 1)(\frac{3N}{2}) & ... & d_4 + (\frac{n}{2} - 1)(\frac{3N}{2}) \\
    d_3 - \gamma(N) &  ...& d_3 - \gamma(N) & d_4- \omega(N)  & ... & d_4- \omega(N) \\
    d_3 - \gamma(N) + n &  ...& d_3 - \gamma(N) + n & d_4- \omega(N) + \frac{3N}{2} & ... & d_4- \omega(N) + \frac{3N}{2} \\
    ...&...&...&...&...&...\\
    {\scriptscriptstyle d_3 - \gamma(N) + n(\frac{n}{2} - 1) }& ...&{\scriptscriptstyle d_3 - \gamma(N) + n(\frac{n}{2} - 1) }& {\scriptscriptstyle d_4- \omega(N) + (\frac{n}{2} - 1)(\frac{3N}{2}) }& ... &{\scriptscriptstyle d_4- \omega(N) + (\frac{n}{2} - 1)(\frac{3N}{2})} \\
    \end{bmatrix}
\end{align*}
\normalsize
where 
\begin{align*}
    \scriptstyle \gamma(N) &\scriptstyle= N^3 - N^2\\
    \scriptstyle \omega(N) &\scriptstyle= N^3  - \frac{N^2}{2}\\
    \scriptstyle d_3 = \lfloor D_N - &\scriptstyle(2\cdot T_{\frac{N}{2}} - (2(\frac{N}{2})^2 - 1))\rfloor\\
    \scriptstyle d_4 = \lfloor D_N - (2(\frac{N}{2})^2) - (4T_{\frac{N}{2}} &\scriptstyle- 2(\frac{N}{2})^2 + 2) -((\frac{N}{2})^2 - N) + (\frac{N}{2} + 1)\rfloor
\end{align*}
\end{lemma}
The previous two lemmas contain a large amount of information and can be hard to parse: the takeaway should be that we can symbolically characterize the reuse distances of temporaries in terms of where in the tree they are created and used. The interested reader can see \textbf{github.com/anon} for more information on how the above functions and relationships were derived, but the exact specifics are not essential to understand the larger picture.

Lemmas \ref{lemma:temp1} and \ref{lemma:temp2}, which consider the reuse distances incurred by temporaries in individual nodes of our tree, will be later combined with tree structure information to create the complete picture of temporary reuse.
\subsection{Matrices A and B}
To understand the behavior of the data in matrices $A$ and $B$, we must first define what it means for a reuse of one such datum to be at $LN$ given that all the accesses to $A$ and $B$ occur in leaves of the tree.
\begin{definition}[Input data reuse levels]
\label{def:level}
When referring to a reuse of data in matrix $A$ or $B$, $LN$ indicates the \textbf{largest} N for which there exists a complete $LN$ call between the data item's use and reuse.
\end{definition}
To help visualize, Figure \ref{fig:tree_a11} contains the tree decomposition of a 4x4 matrix multiplication (L4 call) with the nodes that contain a reference to element $A[1,1]$ (1-indexed: the top-left-most element of matrix $A$) highlighted. There are two L1 reuses, between the two blue nodes and between the two pink nodes, and one L2 reuse, between the second blue node and the first pink node (as there exists a full L2 call between those two). The dashed line boxes indicated addition groups.
\begin{figure}[h]
    \centering
    \includegraphics[width=\columnwidth]{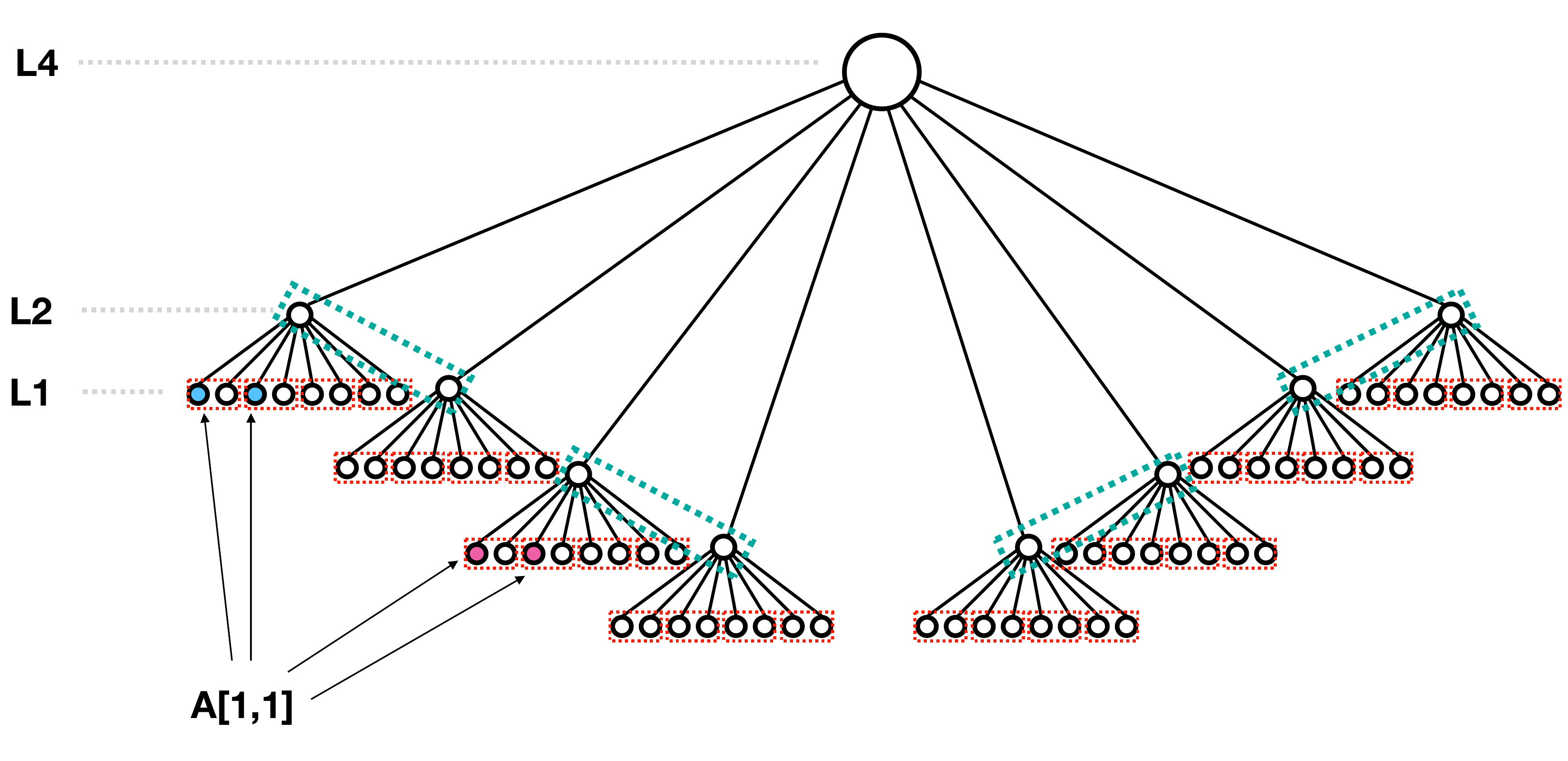}
    \caption{Tree decomposition with A[1,1] accesses highlighted}
    \label{fig:tree_a11}
\end{figure} 
Combining with tree structure, we can count occurrences of $LN$ reuse for all $N$:
\begin{lemma}
\label{lemma:abcount}
Let $RC(N,M)$ denote the number of $LM$ reuses of any item in matrix A or B in an $LN$ call (an NxN multiplication). Then
\begin{align*}
    RC(N,M) = \frac{N}{2M}
\end{align*}
\end{lemma}
With an understanding of the frequency of each type of reuse of data from $A$ and $B$ we can turn our focus to computing the values of the reuse distances. Our high level angle of attack for this will be to again decompose into temporaries vs. input data: we will derive separately the number of temporaries as well as the number of elements of $A$ and $B$ between two consecutive uses of an item from $A$ or $B$. Formally:
\begin{definition}[Reuse distance decomposition]
\label{def:decomp}
Let $F(i, j, N)$ represent the reuse distance of element $A[i,j]$ at an $LN$ reuse. Then
\begin{align*}
    F(i, j, N) = f_T(i,j,N) + f_{A,B}(i,((j-1) \hspace{0.1cm} \% \hspace{0.1cm} N)+1,N),
\end{align*}
where $f_T(i,j,N)$ is the number of unique temporaries within the $A[i,j]$ $LN$ reuse pair, and $f_{A,B}(i,j,N)$ is the number of unique elements of matrices $A,B$ within the same reuse pair. Additionally,
\begin{align*}
    G(i, j, N) = g_T(i,((j-1) \hspace{0.1cm}\%\hspace{0.1cm}N)+1,N) + g_{A,B}(i,j,N),
\end{align*}
for matrix B.
\end{definition}
We will now derive the four summed terms in the Definition \ref{def:decomp} (two in $F$, two in $G$).
\subsubsection{Matrix A}
We again elide the full derivations of $f_T$ and $f_{A,B}$ but point the interested reader to \textbf{github.com/anon}. The results are as follows:
\begin{lemma}[Matrix A temporaries]
\label{lemma:a_temp}
Let $f_T(i,j,N)$ be the number of unique temporaries within an $A[i,j]$ $LN$ reuse pair. Then
\begin{align*}
\begin{aligned}
    f_T(i,j,N) = T_N + 2N^2 + 8T_{\frac{N}{2}} - \sum_{i=0}^{log_2(N) - 1}2T_{2^{k}}\\ + \sum_{k=0}^{log_2(N) + 1} (2^k)^2\cdot I(((j-1 \hspace{0.1cm} \% \hspace{0.1cm} 2N) \hspace{0.1cm} \% \hspace{0.1cm} 2^{k+1}) + 1 > 2^k) \\ + \sum_{k=0}^{log_2(N)} (2^k)^2\cdot I(((i-1 \hspace{0.1cm} \% \hspace{0.1cm} N) \hspace{0.1cm} \% \hspace{0.1cm} 2^{k+1}) + 1 > 2^k)
    \end{aligned}
\end{align*}
\end{lemma}
For $f_{A,B}$, we isolate the following recursive relationship (where the recursion is between parents and children in the tree) as being essential:
\begin{align*}
\begin{bmatrix}
    \delta_1\\\delta_2
\end{bmatrix} 
\rightarrow
\begin{bmatrix}
    \delta_1&\delta_1 + N^2\\N^2 & 2\cdot N^2\\2 \cdot N^2 & N^2\\\delta_2 + N^2 & \delta_2
\end{bmatrix}
\end{align*}
Here, each matrix entry is itself a $\frac{N}{2} \times N$ matrix. Lemma \ref{lemma:a_a} reflects this recursion, with the eight terms corresponding to the eight sections of the second matrix. The indicator function calls isolate which section of the matrix $(i,j)$ fall in, and the four recursive calls correspond to the four sections that contain a $\delta$ term.
\begin{lemma}
\label{lemma:a_a}
Let $f_{A,B}(i,j,N)$ be the number of unique items in matrices $A$ and $B$ within an $A[i,j]$ $LN$ reuse pair.
\begin{align*}
    \begin{aligned}
        f_{A,B}(i,j,N) = 4N^2 + f'_{A,B}(i,j,N)
    \end{aligned}
\end{align*}
where
\begin{align*}
\scriptstyle
    f'_{A,B}(i,j,N) = 4\cdot N^2 + (\frac{N}{2})^2\cdot I(\frac{N}{4} < i \leq \frac{N}{2}, j \leq \frac{N}{2})
    +\hspace{0.1cm} 2(\frac{N}{2})^2\cdot I(\frac{N}{4} < i \leq \frac{N}{2}, \frac{N}{2} < j \leq N)\\
    \scriptstyle +\hspace{0.1cm} (\frac{N}{2})^2\cdot I(\frac{N}{2} < i \leq \frac{3N}{4}, \frac{N}{2} < j \leq N)
    +\hspace{0.1cm} 2(\frac{N}{2})^2\cdot I(\frac{N}{2} < i \leq \frac{3N}{4}, j \leq \frac{N}{2})\\
    \scriptstyle +\hspace{0.1cm} I(i > \frac{3N}{4}, j \leq \frac{N}{2})\cdot((\frac{N}{2})^2 + f_{A,B}(i - \frac{N}{2}, j, \frac{N}{2}))
    +\hspace{0.1cm} I(i > \frac{3N}{4}, j > \frac{N}{2})\cdot(f_{A,B}(i - \frac{N}{2}, j - \frac{N}{2}, \frac{N}{2}))\\
    \scriptstyle +\hspace{0.1cm} I(i \leq \frac{N}{4}, j > \frac{N}{2})\cdot((\frac{N}{2})^2 + f_{A,B}(i, j - \frac{N}{2}, \frac{N}{2}))\\
    \scriptstyle +\hspace{0.1cm} I(i \leq \frac{N}{4}, j \leq \frac{N}{2})\cdot(f_{A,B}(i, j, \frac{N}{2}))\\
\end{align*}
and
\begin{align*}
    f'_{A,B}(i,j,2) = \begin{bmatrix}
    0&1\\2&1
    \end{bmatrix}
\end{align*}
\end{lemma}
\subsubsection{Matrix B}
The quantities of interest for matrix $B$ have very similar forms:
\begin{lemma}[Matrix B temporaries]
\label{lemma:b_temp}
Let $g_T(i,j,N)$ be the number of unique temporaries within an $B[i,j]$ $LN$ reuse pair. Then
\begin{align*}
\begin{aligned}
    g_T(i,j,N) = 4 \cdot T_N + 2N^2 - \sum_{k = 0}^{log_2(N) - 1}4\cdot T_{2^k} + \\\sum_{k=0}^{\lceil log_2(N)\rceil} 4^k \cdot I(((i -1)\hspace{0.1cm} \% \hspace{0.1cm} 2^{k+1}) + 1 > 2^k) \hspace{0.2cm} \\+ \sum_{k=0}^{\lceil log_2(N)\rceil} 4^k \cdot I(((j-1) \hspace{0.1cm}\% \hspace{0.1cm}2^{k+1}) + 1 > 2^k)
    \end{aligned}
\end{align*}
\end{lemma}
Again, we extract the key recursive relationship between reuse distances in parents and children in the tree:
\begin{align*}
\begin{bmatrix}
    \delta_1&\delta_2
\end{bmatrix}
\rightarrow
\begin{bmatrix}
    \delta_1 & 2N^2 & 6N^2 & \delta_2 + (2N)^2\\
    \delta_1 + (2N)^2 & 6N^2 &  2N^2 &\delta_2
    \end{bmatrix}
\end{align*}
where each entry represents an $N \times \frac{N}{2}$ matrix. Lemma \ref{lemma:b_a} again reflects this recursion:
\begin{lemma}
\label{lemma:b_a}
Let $g_{A,B}(i,j,N)$ be the number of unique items in matrices $A$ and $B$ within an $B[i,j]$ $LN$ reuse pair. Then
\begin{align*}
    \begin{aligned}
        g_{A,B}(i,j,N) = 6N^2 + g'_{A,B}(i,j,N)
    \end{aligned}
\end{align*}
where
\begin{align*}
\begin{aligned}
    \scriptstyle  g'_{A,B}(i,j,N) = \frac{N^2}{2}\cdot I(i \leq N, \frac{N}{2} < j \leq N) + \hspace{0.1cm} (\frac{3N^2}{2})\cdot I(i \leq N, N < j \leq \frac{3N}{2}) \\
    \scriptstyle+ \hspace{0.1cm} (\frac{3N^2}{2})\cdot I(N < i \leq 2N, \frac{N}{2} < j \leq N) + \hspace{0.1cm} (\frac{N^2}{2})\cdot I(N < i \leq 2N, N < j \leq \frac{3N}{2}) \\
    \scriptstyle+ \hspace{0.1cm} I(i \leq N, j \leq \frac{N}{2})\cdot g'_{A,B}(i, j, \frac{N}{2})+ \hspace{0.1cm} I(N < i \leq 2N, j \leq \frac{N}{2})\cdot(N^2 + g'_{A,B}(i - N, j, \frac{N}{2}))\\
    \scriptstyle+ \hspace{0.1cm} I(i \leq N, \frac{3N}{2} < j \leq 2N)\cdot(N^2 + g'_{A,B}(i, j - N, \frac{N}{2}))\\\scriptstyle+ \hspace{0.1cm} I(N < i \leq 2N, \frac{3N}{2} < j \leq 2N)\cdot g'_{A,B}(i - N, j - N, \frac{N}{2})\\
    \end{aligned}
\end{align*}
and
\begin{align*}
    g'_{A,B}(i,j,N) = \begin{bmatrix}
    1&1\\2&0
    \end{bmatrix}
\end{align*}
\end{lemma}
\subsection{Final Distribution}
With Definitions 1-4 and Lemmas 1-8, we have an understanding of both the values of the reuse distances that correspond to nodes in our tree decomposition and their frequencies. We then arrive at the complete specification for a reuse distance distribution (i.e. miss ratio curve):
\begin{theorem}[Reuse Distance Multiset]
\label{thm:rds}
The multiset of all reuse distances in an execution of $NxN$ recursive matrix multiplication can be expressed as follows:
\begin{align*}
    \begin{aligned}
        RD_N = \bigcup_{l \in \{1,2,4...N\}}\bigcup_{(i,j,k) \in \{1..L\} \times \{1..L\} \times \{1,2\}} \underbrace{\{DT_{k,l}(i,j)...DT_{k,l}(i,j)\}}_{\frac{\#Ll}{2}}
        \\\cup\hspace{0.2cm}\bigcup_{l \in \{1,2,4...N\}}\bigcup_{(i,j) \in \{1..L\} \times \{1..2L\}}\underbrace{\{F(i,j,l) ... F(i,j,l)\}}_{\frac{N}{2l} \cdot \frac{N^2}{2l^2}}
        \\\cup\hspace{0.2cm}\bigcup_{l \in \{1,2,4...N\}}\bigcup_{(i,j) \in \{1..2L\} \times \{1..2L\}}\underbrace{\{G(i,j,l) ... G(i,j,l)\}}_{\frac{N}{2l}\cdot \frac{N^2}{4l^2}}
    \end{aligned}
\end{align*}
where $DT_{k,l}(i,j), F(i,j,N),G(i,j,N)$ and $\#L_l$ are defined in earlier lemmas and definitions.
\end{theorem}
Theorem \ref{thm:rds} is a symbolic construction that represents RMM's reuse distance behavior for any input size as a multiset. When instantiated for a given execution, this data would normally take the form of a histogram, but as we are specifying a distribution symbolically we require this multiset construction. The elements of the set are reuse distance values, and their repetition denotes the multiplicity of that RD value.
\subsection{Algorithmic form}
Algorithm \ref{alg:rdd} contains a specification for how to compute a reuse distribution for RMM given the previous handful of mathematical results and an input size:
\begin{algorithm}
\caption{Reuse Distance Computation}\label{alg:rdd}
\begin{algorithmic}[1]
\Require $RD : N \rightarrow N$\Comment Dictionary. key:RD, value:count.
\Procedure{compute\_RMM\_RDD}{$N$}
\For{$l \in \{1,2,4...N\}$}
\For{$(i,j,k) \in \{1..L\} \times \{1..L\} \times \{1,2\} $}
\State {\color{blue}/* Lemmas \ref{lemma:symmetry}, \ref{lemma:temp1}, \ref{lemma:temp2}, Definition \ref{def:nodecount}*/}
% \State $RD[DT_{k,l}(i,j)]$ \leftarrow $RD[DT_{k,l}(i,j)] + \frac{\#Ll}{2}$
\State $RD[DT_{k,l}(i,j)] \leftarrow RD[DT_{k,l}(i,j)] + \frac{\#Ll}{2}$
\EndFor
\For{$(i,j) \in \{1..L\} \times \{1..2L\} $}
\State {\color{blue}/* Lemmas \ref{lemma:abcount}, \ref{lemma:a_temp}, \ref{lemma:a_a}, Definition \ref{def:decomp} */}
\State $RD[F(i,j,l)] \leftarrow RD[F(i,j,l)] + \frac{N}{2l} \cdot \frac{N^2}{2l^2}$
\EndFor
\For{$(i,j) \in \{1..2L\} \times \{1..2L\} $}
\State {\color{blue}/* Lemmas \ref{lemma:abcount}, \ref{lemma:b_temp}, \ref{lemma:b_a}, Definition \ref{def:decomp} */}
\State $RD[G(i,j,l)] \leftarrow RD[G(i,j,l)] + \frac{N}{2l} \cdot \frac{N^2}{4l^2}$
\EndFor
\EndFor
\State \textbf{return} $RD$\Comment{Dictionary stores distribution}
\EndProcedure
\end{algorithmic}
\end{algorithm}

Algorithm \ref{alg:rdd} has runtime $O(n^2 \cdot log(n))$ for matrix dimension $n$, while collecting this data by running the program and performing trace analysis has runtime $O(n^3 \cdot log(n))$ \cite{Olken:LBL81}. Any profiling involving running the program must be $\Omega(n^3)$, demonstrating that our approach has guaranteed asymptotic improvement.
\subsection{Verification}
We verified Theorem \ref{thm:rds} by comparing its resultant RD distribution to an RD distribution formed by instrumenting RMM and performing trace analysis. The two distributions are verified to be identical up to size $256\times256$.
\section{Data Movement Distance Analyses}
In the following section we will derive bounds on the data movement distance (see Definition \ref{def:dmd}) incurred by several forms of matrix multiplication: naive, tiled, and recursive and Strassen's both with and without temporary reuse. In doing so, we demonstrate the following two valuable properties of DMD:
\begin{enumerate}
    \item DMD is capable of asymptotically differentiating algorithms with the same time complexity as a result of their memory behavior
    \item DMD is capable of asymptotically differentiating between versions of the same algorithm with and without locality optimizations
\end{enumerate}
We construct precise DMD values for naive MM, asymptotically tight upper and lower bounds (differing only in coefficient) for tiled and recursive MM, and upper bounds for Strassen's algorithm and recursive MM with memory management.
\subsection{Naive Matrix Multiplication}
\begin{comment}
To precisely derive naive MM's DMD, we first must examine its reuse distance distribution.  The left matrix $L$ has the following distribution:
\begin{align*}
   P(\rd = c) =\begin{cases} 1 & c = 2n \\ 0 & otherwise \end{cases}
\end{align*}

The right matrix $R$ has a more complex distribution. First, we will look at reuse distance for element $R[i,j]$ as a function of $i,j$:
\begin{align*}
    \rd(R[i,j]) = \begin{cases} n^2 + n + i & j = 1 \\ n^2 + 2n - i & j = n \\ n^2 + 2n & otherwise
    \end{cases}
\end{align*}
This variance between $n^2 +n$ and $n^2 + 2n$ is because elements in the first and last columns of $R$ see all elements of $R$, all elements in one row or $L$, and part of another row of $L$ in between uses. In contrast, elements outside the first and last row see all elements of $R$ and all elements in two rows of $L$. The resulting reuse distance distribution then looks as follows:

The right matrix $R$:
\begin{align*}
   P(\rd = c) =\begin{cases} \frac{n^2 - 2n + 1}{n^2} & c = n^2 + 2n \\ \frac{1}{n^2} & c = n^2 + n
   \\ \frac{2}{n^2} & c = n^2 + n + 1
   \\ \frac{2}{n^2} & c = n^2 + n + 2
   \\ ... & ...
   \\ \frac{2}{n^2} & c = n^2 + 2n - 1
   \end{cases}
\end{align*}
\end{comment}
For space, we elide the derivation of naive MM's DMD. However, it is quite straightforward. We derive first is reuse distance distribution, which is very simple, then apply Defition \ref{def:dmd} to sum the square roots of all its RDs. The result:
\begin{align*}
    DMD_{MM} = (n^3 \cdot \sqrt{2n}) + (n^3 - 2n^2 + n) \cdot \sqrt{n^2 + 2n}  \\+ (\sum_{i = 1}^{n - 1} 2n \cdot \sqrt{n^2 + n + i}) + (n \cdot \sqrt{n^2 + n})
\end{align*}
Simplifying asymptotically:
\begin{align*}
    DMD_{MM} \sim n^4
\end{align*}

\subsection{Tiled Matrix Multiplication}
Consider matrix multiplication with the computation reordered by partitioning the input matrices into $DxD$ tiles as follows, from \cite{BaoD:CGO13}:
\begin{verbatim}
    for(jj = 0; jj < N; jj = jj + D)
        for(kk = 0; kk< N; kk = kk + D)
            for(i = 0; i< N; i = i + 1)
                for(j = jj; j < jj + D; j = j + 1)
                    for(k = kk; k < kk + D; k = k + 1)
                        C[i][j] = C[i][j] + A[i][k]*B[k][j]
\end{verbatim}
For space, we elide the full derivation of tiled MM's DMD, but the approach is similar to that for naive matrix multiplication above. However, instead of deriving the precise reuse distance distribution, we form upper and lower bounded versions and use them to derive the corresponding DMD bounds. The interested reader can see the full derivation at \textbf{github.com/anon}.
\begin{theorem}[Tiled Matrix Multiplication DMD]
Upper and lower bounds on the data movement distance incurred by tiled matrix multiplication operating on $NxN$ matrices with $DxD$ tiles are as follows:
\begin{align*}
    \frac{N^4}{D} + N^3\cdot D < \sim(DMD_{TMM}) < 2\sqrt{3}\frac{N^4}{D} + \sqrt{2}N^3\cdot D
\end{align*}
\end{theorem}
Note that for $D=1$ and $D=N$, the lower bound is exactly the asymptotic DMD of naive matrix multiplication, as tile sizes of 1 and $N$ result in the same computation order as naive MM.

\subsection{Recursive Matrix Multiplication}
\textbf{$F:$}
\medskip \newline Firstly, note that as $f_T(i,j,N) = \Omega(N^3)$ and $f_{A,B}(i,j,N) = O(N^2)$, $F(i,j,N) \sim f_T(i,j,N)$ (see Definition \ref{def:decomp}). Similarly, $G(i,j,N) \sim g_T(i,j,N)$. Noting this:
\begin{align*}
    f_T(i,j,N) \sim T_N + 8T_{\frac{N}{2}} - \sum_{k=0}^{log_2(N) - 1}2\cdot T_{2^{k}}
\end{align*}
Taking a lower bound on the summation to derive an upper bound on DMD:
\begin{align*}
    F(i,j,N) \sim 3N^3
\end{align*}
Upper bounding the summation for a lower bound:
\begin{align*}
    F(i,j,N) \sim 2N^3
\end{align*}
$G$ has a similar analysis:
\begin{align*}
    g_T(i,j,N) \sim 4\cdot T_N - \sum_{k=0}^{log_2(N) - 1}2\cdot T_{2^{k}}
\end{align*}
An upper bound here, bounding the summation again:
\begin{align*}
    G(i,j,N) \sim 7\cdot N^3
\end{align*}
The corresponding lower bound:
\begin{align*}
    G(i,j,N) \sim 6\cdot N^3
\end{align*}
\subsection{Accesses to temporaries}
We now turn our attention to accesses to temporaries. We will consider the four quadrants of  $DT_{1,N}, DT_{2,N}$ separately. First, we need the asymptotic behavior of $d_1, d_2, d_3, d_4, \phi, \delta, \omega, \gamma$:
\begin{align*}
    d_1 \sim \frac{7N^3}{2}, d_2 \sim 3N^3, d_3 \sim \frac{3N^3}{2}, d_4 \sim N^3\\
    \phi(N), \gamma(N), \omega(N), \delta(N) \sim N^3
\end{align*}
The distribution of asymptotic reuse distances for $DT_1$ and $DT_2$ are then as follows, by partitioning each into four quadrants and combining the asymptotic costs of the above:

\begin{align*}
    DT_{1,N} : \left\{
        \begin{array}{ll}
            1/4 & \frac{7N^3}{2}  \vspace{0.1cm} \\
            1/4 & \frac{5N^3}{2} \vspace{0.1cm}\\
            1/4 & 3N^3 \vspace{0.1cm}\\
            1/4 & 2N^3 
        \end{array}
    \right\|
    \hspace{1.0cm} DT_{2,N} : \left\{
        \begin{array}{ll}
            1/4 & \frac{3N^3}{2}  \vspace{0.1cm} \\
            1/4 & \frac{N^3}{2} \vspace{0.1cm}\\
            1/4 & N^3 \vspace{0.1cm}\\
            1/4 & 3N^2 
        \end{array}
    \right\|
\end{align*}
\subsection{DMD calculation} 

With asymptotic reuse distance for each type of memory access, we can now use the frequency information in Theorem \ref{thm:rds} to calculate DMD.
\subsubsection{Temporaries}
$DT_1:$
\begin{align*}
    DMD_{DT1} = \sum_{i=0}^{log_2(N)}\frac{N^3\cdot\sqrt{2(2^i)^3}}{4\cdot2^i} + \sum_{i=0}^{log_2(N)}\frac{N^3\cdot\sqrt{3(2^i)^3}}{4\cdot2^i} \\+ \sum_{i=0}^{log_2(N)}\frac{N^3\cdot\sqrt{\frac{7(2^i)^3}{2}}}{4\cdot2^i} + \sum_{i=0}^{log_2(N)}\frac{N^3\cdot\sqrt{\frac{5(2^i)^3}{2}}}{4\cdot2^i}\\
DMD_{DT1} \sim \frac{2 + \sqrt{5} + \sqrt{6} + \sqrt{7} + 2\sqrt{2} + 2\sqrt{3} + \sqrt{14} + \sqrt{10} }{4}N^{3.5}
\end{align*}
$DT_2:$
\begin{align*}
    DMD_{DT2} = \sum_{i=0}^{log_2(N)}\frac{N^3\cdot\sqrt{(2^i)^3}}{4\cdot2^i} + \sum_{i=0}^{log_2(N)}\frac{N^3\cdot\sqrt{\frac{(2^i)^3}{2}}}{4\cdot2^i} \\+ \sum_{i=0}^{log_2(N)}\frac{N^3\cdot\sqrt{\frac{3(2^i)^3}{2}}}{4\cdot2^i} + \sum_{i=0}^{log_2(N)}\frac{N^3\cdot\sqrt{3(2^i)^2}}{4\cdot2^i}\\
    DMD_{DT2} \sim \frac{3 + 2\sqrt{2} + \sqrt{3} + \sqrt{6}}{4}N^{3.5}\hspace{1.2cm}
\end{align*}
\subsubsection{Accesses to A and B}
First, we will consider the upper bound on $F$:
\begin{align*}
DMD^{up}_{F} = \sum_{i=0}^{log_2(N)}\sum_{j=1}^{2^i}\sum_{k=1}^{2^{i+1}}\frac{N\cdot\sqrt{3(2^i)^3}}{2\cdot2^i} = \sum_{i=0}^{log_2(N)} N \cdot 2^i \sqrt{3(2^i)^3}
\end{align*}
\begin{align*}
    DMD^{up}_{F} \sim \frac{4 \sqrt{6}}{4 \sqrt{2} - 1} N^{3.5}
\end{align*}
Next, the lower bound:
\begin{align*}
DMD^{low}_{F} = \sum_{i=0}^{log_2(N)}\sum_{j=1}^{2^i}\sum_{k=1}^{2^{i+1}}\frac{N\cdot\sqrt{2(2^i)^3}}{2\cdot2^i} = \sum_{i=0}^{log_2(N)} N \cdot 2^i \sqrt{2(2^i)^3}
\end{align*}
\begin{align*}
    DMD^{low}_{F} \sim \frac{8}{4 \sqrt{2} - 1} N^{3.5}
\end{align*}
Similarly for B:
\begin{align*}
DMD^{up}_{G} = \sum_{i=0}^{log_2(N)}\sum_{j=1}^{2^{i+1}}\sum_{k=1}^{2^{i}}\frac{N\cdot\sqrt{7(2^i)^3}}{2\cdot2^i} = \sum_{i=0}^{log_2(N)} N \cdot 2^i \sqrt{7(2^i)^3}
\end{align*}
\begin{align*}
    DMD^{up}_{G} \sim \frac{4 \sqrt{14}}{4 \sqrt{2} - 1} N^{3.5}
\end{align*}
Next, the lower bound:
\begin{align*}
DMD^{low}_{G} = \sum_{i=0}^{log_2(N)}\sum_{j=1}^{2^{i+1}}\sum_{k=1}^{2^{i}}\frac{N\cdot\sqrt{6(2^i)^3}}{2\cdot2^i} = \sum_{i=0}^{log_2(N)} N \cdot 2^i \sqrt{6(2^i)^3}
\end{align*}
\begin{align*}
    DMD^{low}_{G} \sim \frac{8\sqrt{3}}{4 \sqrt{2} - 1} N^{3.5}
\end{align*}
\subsubsection{Total DMD}
\begin{theorem}[Recursive Matrix Multiplication Data Movemement Distance]
Combining the previous, upper and lower bounds on the data movement distance incurred by a standard recursive matrix multiplication algorithm on $NxN$ matrices is as follows:
\begin{align*}
    12.82N^{3.5} < \sim(D_{RMM}(N)) < 13.46N^{3.5}
\end{align*}

\end{theorem}
\subsubsection{Memory Management}
The RMM pseudocode analyzed earlier allocates memory on each call but contains no calls to \textbf{free()}, resulting in poor locality. Practical implementations of RMM will use memory management strategies for handling temporaries, so we will now adapt the previous DMD analysis to take this into account.

One way to introduce temporary freeing would be to call \textbf{free()} twice after each addition group, once the addition result has been stored in a $C$ submatrix. Doing so creates the following upper bound on the total number of temporaries needed:
\begin{align*}
    \#T(N) = N^2 + \sum_{i=0}^{log_2(N) - 1} 2*(2^i)^2 = \frac{2}{3}(N^2 - 1)\\
    \#T(N) < 2N^2\hspace{2.0cm}
\end{align*}
With this, a bound on the total data size of the execution (including input data) is $N^2 + N^2 + 2N^2 = 4N^2$. We previously derived reuse distances as functions of tree position: we can now reuse our analysis of the DMD of RMM without temporary reuse, but instead of using said functions, we will use the minimum of those functions and $4N^2$ (as reuse distance is always upper bounded by data size). For temporary matrix $DT_1$:
\begin{align*}
    &DMD^{up}_{DT1} = \\&\sum_{i=0}^{log_2(N)}\frac{N^3\cdot\sqrt{min(2(2^i)^3, 4N^2)}}{4\cdot2^i} + \sum_{i=0}^{log_2(N)}\frac{N^3\cdot\sqrt{min(3(2^i)^3, 4N^2)}}{4\cdot2^i} \\&+ \sum_{i=0}^{log_2(N)}\frac{N^3\cdot\sqrt{min(\frac{7(2^i)^3}{2}, 4N^2)}}{4\cdot2^i} + \sum_{i=0}^{log_2(N)}\frac{N^3\cdot\sqrt{min(\frac{5(2^i)^3}{2}, 4N^2)}}{4\cdot2^i}
\end{align*}
Solving for the points where the parameters to \textbf{min()} are equal, we partition each summation into two by splitting values of induction variable \textit{i}:
\begin{align*}
   \textstyle DMD^{up}_{DT1} = \sum_{i=0}^{\lfloor\frac{log_2(2N^2)}{3}\rfloor}\frac{N^3\cdot\sqrt{2(2^i)^3}}{4\cdot2^i} + \sum_{i = \lfloor\frac{log_2(2N^2)}{3}\rfloor + 1}^{log_2(N)} \frac{N^4}{2(2^i)} \\\textstyle+ \sum_{i=0}^{\lfloor\frac{log_2(\frac{4}{3}N^2)}{3}\rfloor}\frac{N^3\cdot\sqrt{3(2^i)^3}}{4\cdot2^i} + \sum_{\lfloor\frac{log_2(\frac{4}{3}N^2)}{3}\rfloor + 1}^{log_2(N)} \frac{N^4}{2(2^i)} \\\textstyle+ \sum_{i=0}^{\lfloor\frac{log_2(\frac{8}{7}N^2)}{3}\rfloor}\frac{N^3\cdot\sqrt{\frac{7}{2}(2^i)^3}}{4\cdot2^i} + \sum_{\lfloor\frac{log_2(\frac{8}{7}N^2)}{3}\rfloor + 1}^{log_2(N)} \frac{N^4}{2(2^i)} \\\textstyle+ \sum_{i=0}^{\lfloor\frac{log_2(\frac{8}{5}N^2)}{3}\rfloor}\frac{N^3\cdot\sqrt{\frac{5}{2}(2^i)^3}}{4\cdot2^i} + \sum_{\lfloor\frac{log_2(\frac{8}{5}N^2)}{3}\rfloor + 1}^{log_2(N)} \frac{N^4}{2(2^i)}
\end{align*}
Evaluating:
\begin{align*}
    DGC_{DT1} \sim (\frac{1}{2\cdot 2^{\frac{1}{3}}} + \frac{1}{2\cdot \frac{4}{3}^{\frac{1}{3}}} + \frac{1}{2\cdot \frac{8}{7}^{\frac{1}{3}}} + \frac{1}{2\cdot \frac{8}{5}^{\frac{1}{3}}} \\+ \frac{(2 + \sqrt{2})}{4}(2^{\frac{1}{6}}\cdot \sqrt{2} + \frac{4}{3}^{\frac{1}{6}}\cdot \sqrt{3} + \frac{8}{7}^{\frac{1}{6}}\cdot \sqrt{\frac{7}{2}} + \frac{8}{5}^{\frac{1}{6}}\cdot \sqrt{\frac{5}{2}}))N^{\frac{10}{3}}
\end{align*}
The same analysis for $DT_2$ results in the following:
\begin{align*}
    DMD^{up}_{DT2} \sim (\frac{1}{2\cdot 4^{\frac{1}{3}}} + \frac{1}{2\cdot 8^{\frac{1}{3}}} + \frac{1}{2\cdot \frac{8}{3}^{\frac{1}{3}}} \\+ \frac{(2 + \sqrt{2})}{4}(2^{\frac{1}{3}} + 8^{\frac{1}{6}}\cdot \sqrt{\frac{1}{2}} + \frac{8}{3}^{\frac{1}{6}}\cdot \sqrt{\frac{3}{2}}))N^{\frac{10}{3}}
\end{align*}
We analyze $F$ and $G$ in the same way, but find that they are asymptotically insignificant in comparison to $DT_1$ and $DT_2$. So, we arrive at our DMD bound for RMM with memory management:
\begin{theorem}[RMM Data Movemement Distance With Temporary Reuse: Upper Bound]
Combining the previous, an upper bound on the data movement distance incurred by a  recursive matrix multiplication algorithm on $NxN$ matrices employing temporary reuse is as follows:
\begin{align*}
    \sim(DGC^{up}_{RMM}) < 11.85N^{\frac{10}{3}}
\end{align*}
\end{theorem}

\begin{table*}[t!]
\centering
\begin{tabular}{ccccccc}
\toprule
\multirow{2}{*}{\textbf{Algorithm}}      & \multicolumn{2}{c}{\textbf{MM}} & \multicolumn{2}{c}{\textbf{RMM}} & \multicolumn{2}{c}{\textbf{Strassen}} \\ \cmidrule(l){2-7}
                             &  Naive & Tiled & Naive & Temporary Reuse & Naive & Temporary Reuse  \\ \midrule
                             \textbf{Time Comp.} & $O(N^3)$ & $O(N^3)$ & $O(N^3)$ & $O(N^3)$ & $O(N^{2.8})$ & $O(N^{2.8})$ \\\midrule
                             \textbf{DMD} & $N^4$ & $\sqrt{2}N^3D + \frac{2\sqrt{3}N^4}{D}$ & $13.46N^{3.5}$ & $11.85N^{3.33}$ &  $6.51N^{3.4}$ &$15.36N^{3.23}$\\
                             
                             \bottomrule
\end{tabular}
\caption{Summary of DMD complexity compared to time complexity}
\label{tab:dmd_results}
\end{table*}

\subsection{Strassen's Algorithm}
Strassen's algorithm for matrix multiplication, which reduces time complexity from $O(N^3)$ to around $O(N^{2.8})$ at the cost of worse locality and additional $O(N^2)$ cost, is our final target for analysis. Figure \ref{fig:strassen} \cite{wiki_Strassen} gives pseudocode for Strassen's, which is similar to standard RMM. The core idea is that each call to Strassen's decomposes into 7 recursive calls instead of 8 but contains additional matrix arithmetic. As before, we will analyze Strassen's both with and without locality optimizations for temporary reuse. 
\begin{figure}[h]
    \centering
    \includegraphics[width=\columnwidth]{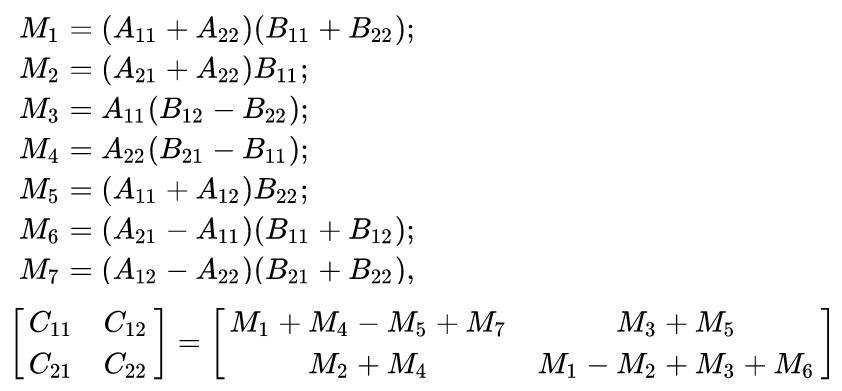}
    \caption{Strassen's algorithm pseudocode}
    \label{fig:strassen}
\end{figure}
Note in the pseudocode that, at each level of recursion, 17 temporary matrices are introduced: $M_1...M_7$ store the results of recursive computation, and there are ten matrix additions or subtractions that are then passed into recursive calls. As each of these matrices is $\frac{N}{2} \times \frac{N}{2}$, the total number of temporaries needed for this call is $\frac{17N^2}{4}$. Summing over all nodes in the tree decomposition:
\begin{lemma}[Strassen's Algorithm Temporary Usage]
    Let $TS_N$ be the total number of temporaries required by an $N \times N$ execution of Strassen's algorithm where all temporaries are unique. Then
    \begin{align*}
        TS_N = \sum_{i=1}^{log_2(N)} \frac{17}{4} \cdot (2^i)^2 \cdot 7^{log_2(N) - i} = \frac{17}{3}(N^{log_2(7)} - N^2)
    \end{align*}
\end{lemma}
We can now asymptotically upper bound the reuse distances of each access from each quadrant of $A,B$ as well as the temporary matrices $M$ in terms of $TS_N$ by counting the calls to an $\frac{N}{2} \times \frac{N}{2}$ matrix multiplication between them. For example, in Figure \ref{fig:strassen} we see that there exist three calls between the first and second uses of $A_{1,1}$, upper bounding each reuse distance from an access in $A_{1,1}$ by $3\cdot TS_{\frac{N}{2}}$. Without enumerating all of these distances, we see there are a total of 31 such intervals, with a total distance of $96 \cdot TS_{\frac{2}{2}}$. As before, we can sum over all nodes to compute DMD:
\begin{align*}
    DMD_{Strassen} \leq \sum_{i=1}^{log_2(N)} \sqrt{78 \cdot TS_{2^{i-1}}} \cdot \frac{(2^i)^2}{4} \cdot 7^{log_2(N) - i}
\end{align*}
Evaluating and asymptotically simplifying:
\begin{align*}
    \sim (DMD_{Strassen}) \leq \frac{4\sqrt{34}(N^2 \sqrt{N^{log_2(7)}} - N^{log_2(7)})}{4\sqrt{7} - 7} 
\end{align*}
\begin{theorem}[Strassen's Algorithm DMD: Upper Bound]
An upper bound on the data movement distance incurred by Strassen's algorithm operating on $NxN$ matrices without temporary reuse is as follows:
\begin{align*}
    \sim (DMD_{Strassen}^{up}) < 6.51 N^{(2 + \frac{log_2(7)}{2})}
\end{align*}
\begin{align*}
    \sim (DMD_{Strassen}^{up}) <\approx 6.51 N^{3.4}
\end{align*}
\end{theorem}

\subsubsection{Memory Management}
We will adapt the previous analysis in the same way that we did when exploring the effect of adding temporary reuse to RMM. First, we note that Huss-Lederman et al. \cite{Huss-Lederman} discuss a memory management technique for Strassen's that requires only $N^2$ temporaries. This brings total data size for an execution to $3N^2$. We will again use this data size as an upper bound on the value of an individual reuse distance.

The largest of the intervals in the pseudocode is $7\cdot TS_{\frac{N}{2}}$, so we replace each of them with this so there is only one intersection point to consider. That shifts the sum from $96\cdot TS_{\frac{N}{2}}$ to $31 \cdot 7\cdot TS_{\frac{N}{2}} = 217 \cdot TS_{\frac{N}{2}}$. We must multiply our $3N^2$ upper bound by 31 as well, yielding $93N^2$. 

\begin{align*}
    \hspace{-0.5cm}DMD_{Strassen'}^{up} = \sum_{i=1}^{log_2(N)} \sqrt{min(217\cdot TS_{2^{i-1}}, 93N^2)} \cdot \frac{(2^i)^2}{4} \cdot 7^{log_2(N) - i}
\end{align*}
Again we solve for an approximate equivalence point for the two functions that preserves the upper bound:
\begin{align*}
    i \approx log_2(0.797 N^\frac{2}{log_2(7))})
\end{align*}

Partitioning the previous summation:
\medskip \newline \begin{align*}
    DMD_{Strassen'}^{up} = \sum_{i=1}^{\lfloor log_2(0.797 N^\frac{2}{log_2(7))})\rfloor} \sqrt{217\cdot TS_{2^{i-1}}} \cdot \frac{(2^i)^2}{4} \cdot 7^{log_2(N) - i}\\
    + \sum_{i=\lceil log_2(0.797 N^\frac{2}{log_2(7))})\rceil}^{log_2(N)} \sqrt{93N^2} \cdot \frac{(2^i)^2}{4} \cdot 7^{log_2(N) - i}
\end{align*}
Simplifying (while preserving the upper bound), removing asymptotically insignificant terms, and evaluating:
\begin{comment}
\begin{align*}
    DMD_{Strassen'}^{up} = 9.79\cdot 7^{log_2(N) - log_2(0.797\cdot N^\frac{2}{log_2(7)})}(0.462\cdot N \cdot N^{\frac{4}{log_2(7)}} \\ - 7^{log_2(0.797\cdot N^\frac{2}{log_2(7)})})
    + \frac{3.22 \cdot N \cdot (2.11 N^{4.23} - N^4)}{N^2}
\end{align*}
Simplifying and removing asymptotically insignificant terms:

As before, we partition the summation on this equivalence point, simplify, and remove asymptotically insignificant terms:
\end{comment}
\begin{theorem}[Strassen's Algorithm With Temporary Reuse DMD: Upper Bound]
An upper bound on the data movement distance incurred by Strassen's algorithm operating on $NxN$ matrices with temporary reuse is as follows:
\begin{align*}
    \sim (DMD_{Strassen'}^{up}) < 15.36 N^{3.23}
\end{align*}
\end{theorem}
\subsection{Summary}
% The following uses of Strassen was pointed out by Sree.
Table \ref{tab:dmd_results} contains asymptotic simplifications of the results of the previous DMD analyses, with a mix of precise results and upper bounds. We make the following observations:
\begin{itemize}
    \item DMD is able to distinguish both between different algorithms with the same time complexity and between locality optimized vs. non-optimized versions of the same algorithm,
    \item The DMD reduction ($\approx N^{\frac{1}{6}}$) from adding temporary reuse is the same for RMM and Strassen even though they have different time and space complexities without temporary reuse,
    \item The gap between Strassen and RMM DMD is smaller than the gap between their time complexities, demonstrating that DMD has captured some of the factors that make Strassen not practical.
\end{itemize}

%The DMD ratio of recursive tiling to naive is at least $\frac{12N^{\frac{10}{3}}}{N^4}$, and 
%Strassen to recursive is at most $\frac{17}{12 N^{0.1}}$.  When $N=1024$, the recursive algorithm DMD is at most around a tenth (12\%) of naive, and Strassen at least 71\% of recursive.  When $N=2048$, the ratios are 7\% and 66\%.

\section{Related Work}
\citet{HongK:STOC81} pioneered the study of I/O complexity, measuring memory complexity by the amount of data transfer and deriving this complexity symbolically as a function of the memory size and the problem size.  The same complexity measures were used in the study of cache oblivious algorithms~\citep{Frigo+:FOCS99} and communication-avoiding algorithms. Olivry et al. \cite{olivry+PLDI20} introduce a compiler technique to statically derive I/O complexity bounds. Olivry et al. use asymptotic complexity ($\sim$), much as we do, to consider constant-factor performance differences, while the rest do not. I/O complexity suffers from an issue inherited from miss ratio curves: it is not ordinal across cache sizes. Miss ratio as a function of cache size can be numerically compared by an algorithm designer when a target cache size is known, but it cannot be used to evaluate the general effectiveness of an optimization across cache sizes. DMD, on the other hand, is an ordinal metric that is agnostic to cache size.

Memory hierarchies in practice may vary in many ways, which make a unified cost model difficult.  \citet{Valiant:ESA08} defined a bridging model, Multi-BSP, for a multi-core memory hierarchy with a set of parameters including the number of levels and the memory size and three other factors at each level. A simpler model was the uniform memory hierarchy (UMH) by \citet{Alpern+:UMH94} who used a single scaling factor for the capacity and the access cost across all levels.  Both are models of memory, where caching is not considered beyond the point that the memory may be so implemented.
\begin{comment}
I/O complexity focuses on communication volume, and Multi-BSP on communication cost.  They are defined at a level above cache management, and their result is independent of cache performance.  In comparison, this work focuses on the distance or cost of data movement and the effect of caching. 
\end{comment}

Matrix multiplication is well researched. Much effort has been put into the analysis and optimization of the Strassen algorithm and its cache utilization \cite{Thottethodi+ISC98, Pauca+98, Huss-Lederman, Singh_comparativestudy} as well as its parallel behavior \cite{Tang+SPAA20, Blelloch+SODA08}. Lincoln et al. \cite{Lincoln+SPAA18} explore performance in a dynamically sized caching environment. Recently, researchers have argued that, with proper implementation considerations and under the correct conditions, the Strassen algorithm can outperform more conventionally practical variants of MM \cite{Huang+SC16}.

\section{Conclusion}
We have explored the interactions between six variants of matrix multiplication and hierarchical memory through the lens of Data Movement Distance. We have demonstrated that DMD's assumptions conform with microarchitectural trends and that it is capable of exposing algorithmic properties that traditional analyses cannot. Time complexity, while an important theoretic metric by which to analyze algorithms, is at odds with a computing environment in which memory systems are increasingly large and complex. We argue that data movement complexity analysis through DMD has great potential to help engineers and algorithm designers to understand the algorithmic implications of hierarchical memory.

% \begin{acks}

% \end{acks}

\balance
\setcitestyle{numbers} 
\bibliographystyle{ACM-Reference-Format}
\bibliography{all}

\appendix

\end{document}